\documentclass[a4paper,USenglish,cleveref, autoref, thm-restate, pdfa
]{lipics-v2021}

\hideLIPIcs

\title{Evidence for Long-Tails in SLS Algorithms} %

\titlerunning{Evidence for Long-Tails in SLS Algorithms} %

\author{Florian Wörz}{Universität Ulm, Institut für Theoretische Informatik, Ulm, Germany \and \url{https://www.uni-ulm.de/in/theo/m/woerz/} }{florian.woerz@uni-ulm.de}{https://orcid.org/0000-0003-2463-8167}{Supported by the Deutsche Forschungsgemeinschaft (DFG) under project number 430150230, ``Complexity measures for solving propositional formulas''.}%

\author{Jan-Hendrik Lorenz}{Universität Ulm, Institut für Theoretische Informatik, Ulm, Germany}{jan-hendrik.lorenz@uni-ulm.de}{https://orcid.org/0000-0002-9554-4347}{}

\authorrunning{F. Wörz and J.-H. Lorenz} %

\Copyright{Florian Wörz and Jan-Hendrik Lorenz} %



\category{} %

\supplement{See \cite{AllData}:}%
\supplementdetails[linktext={}, subcategory={}, swhid={}]{Base instances and modifications}{https://doi.org/10.5281/zenodo.4715893} %
\supplementdetails[linktext={}, subcategory={}, swhid={}]{Visual and statistical evaluations}{https://doi.org/10.5281/zenodo.5026180} %

\acknowledgements{The authors acknowledge support by the state of Baden-Württemberg through bwHPC.}%

\nolinenumbers %

\EventEditors{Petra Mutzel, Rasmus Pagh, and Grzegorz Herman}
\EventNoEds{3}
\EventLongTitle{29th Annual European Symposium on Algorithms (ESA 2021)}
\EventShortTitle{ESA 2021}
\EventAcronym{ESA}
\EventYear{2021}
\EventDate{September 6--8, 2021}
\EventLocation{Lisbon, Portugal}
\EventLogo{}
\SeriesVolume{204}
\ArticleNo{10}

\newcommand{\ie}{i.\,e.,\ }
\newcommand{\eg}{e.\,g.,\ }
\newcommand{\egcite}{e.\,g.}

\newcommand{\iid}{i.\,i.\,d.\ }

\newcommand{\introduceterm}[1]{{\emph{#1}}}
\newcommand{\algoformat}[1]{\texttt{#1}}

\newcommand{\R}{\mathbb{R}}
\newcommand{\Rpos}{\R^{+}}

\newcommand{\Indikator}[1]{\mathbbm{1}_{#1}}

\newcommand{\D}{\textnormal{d}}

\newcommand{\set}[1]{\{ #1 \}}
\newcommand{\Set}[1]{\big\{ #1 \big\}}

\newcommand{\setdescr}[3][\mid]{\set{ #2 #1 #3 }}
\newcommand{\Setdescr}[3][\bigm\vert]{\Set{ #2 #1 #3 }}

\newcommand{\problemlanguageformat}[1]{\textrm{#1}\xspace}
\newcommand{\SAT}{\ensuremath{\problemlanguageformat{SAT}}}

\newcommand{\complexityclassformat}[1]{\textrm{\upshape{\textsf{#1}}}\xspace}
\newcommand{\NP}{\ensuremath{\complexityclassformat{NP}}}

\DeclareMathOperator{\Probop}{\mathrm{Pr}}
\DeclareMathOperator{\Expop}{\mathrm{E}}

\newcommand{\prob}[2][]{\Probop_{#1} \left[ #2 \right]}

\newcommand{\expectation}[2][]{\Expop_{#1} [ #2 ]}

\newcommand{\integral}[4][x]{\int\limits_{#2}^{#3}#4 \, \D #1}
\newcommand{\integralLow}[4][x]{\int_{#2}^{#3}#4 \, \D #1}

\newcommand{\refDef}[1]{Definition~\ref{#1}}
\newcommand{\refLem}[1]{Lemma~\ref{#1}}

\newcommand{\refTheo}[1]{Theorem~\ref{#1}}
\newcommand{\refCon}[1]{Conjecture~\ref{#1}}
\newcommand{\refAlg}[1]{Algorithm~\ref{#1}}

\newcommand{\refFig}[1]{Figure~\ref{#1}}
\newcommand{\refTab}[1]{Table~\ref{#1}}

\newcommand{\refEq}[1]{Equation~\eqref{#1}}

\newcommand{\Clauses}{\ensuremath{L}}
\newcommand{\Alfa}{\algoformat{Alfa}}

\newcommand{\probSAT}{\algoformat{probSAT}}
\newcommand{\SRWA}{\algoformat{SRWA}}
\newcommand{\YalSAT}{\algoformat{YALSAT}}
\newcommand{\YAL}{\YalSAT}
\newcommand{\GapSAT}{\algoformat{GapSAT}}

\newcommand{\Lim}[1]{\lim_{#1 \rightarrow \infty}}
\newcommand{\LIM}[1]{\lim\limits_{#1 \rightarrow \infty}}

\newcommand{\quantorsep}{\,\,}

\usepackage{tikz}
\usepackage{pgfplots}
\usepackage{pgfplotstable}
\usetikzlibrary{patterns,snakes,arrows,shapes,shapes.arrows}
\pgfplotsset{compat=newest}
\usetikzlibrary{external}
\usepackage{bbm}
\usepackage{mathtools}
\usepackage[ruled,noend,vlined]{algorithm2e} %
\usepackage{doi}
\usepackage[skip=-7pt]{caption}

\newtheorem{restatelemmaorig}[theorem]{Lemma}
\newtheorem*{restatelemma}{Lemma \ref{lem:restateorig} Restated in Stronger Form}

\begin{document}

\maketitle

\begin{abstract}
	Stochastic local search (SLS) is a successful paradigm for solving the satisfiability problem of propositional logic.  A recent development in this area involves solving not the original instance, but a modified, yet logically equivalent one~\cite{LW20OnTheEffectOfLearnedClauses}. Empirically, this technique was found to be promising as it improves the performance of state-of-the-art SLS solvers.
	
	Currently, there is only a shallow understanding of how this modification technique affects the runtimes of SLS solvers. Thus, we model this modification process and conduct an empirical analysis of the hardness of logically equivalent formulas. Our results are twofold. First, if the modification process is treated as a random process, a lognormal distribution perfectly characterizes the hardness; implying that the hardness is long-tailed.
	This means that the modification technique can be further improved by implementing an additional restart mechanism.
	Thus, as a second contribution, we theoretically prove that all algorithms exhibiting this long-tail property can be further improved by restarts. Consequently, all SAT solvers employing this modification technique can be enhanced.	
\end{abstract}

\section{Introduction}

Although algorithms for solving the $\NP$-complete satisfiability problem, so-called SAT solvers, are nowadays remarkably successful in solving large instances, randomized versions of these solvers often show a high variation in the runtime required to solve a fixed instance over repeated runs~\cite{GSCK00HeavyTailedPhenomena}.
In the past, research on randomized algorithms often focused on studying the unsteady behavior of statistical measures like the mean, variance, or higher moments of the runtime over repeated runs of the respective algorithm. In particular, these measures are unable to capture the long-tailed behavior of difficult instances.
In a different line of work~\cite{FRV97SummarizingCSPHardness,GS97AlgorithmPortfolioDesign,RF97StatisticalAnalysis}, the focus has shifted to studying the runtime distributions of search algorithms, which helps to understand these methods better and draw meaningful conclusions for the design of new algorithms.

Recently, the hybrid solver \algoformat{GapSAT}~\cite{LW20OnTheEffectOfLearnedClauses} was introduced, combining a stochastic local search~(SLS) solver with a conflict-driven clause learning~(CDCL) solver.
In the analysis conducted,
it was empirically shown that adding new clauses is beneficial to the mean runtime (in flips) of the SLS solver \probSAT{}~\cite{BalintImplementationOfProbSAT} underlying the hybrid model.
The authors also demonstrated that although adding new clauses can improve the mean runtime, there exist instances where adding clauses can harm the performance of SLS.
This behavior is worth studying to help eliminate the risk of increasing the runtime of such procedures.

For this reason, we study the runtime (or more precisely, hardness) distribution of the procedure~\Alfa{}, introduced in this work, that models the addition of a set of logically equivalent clauses~$\Clauses$ to a formula~$F$ and the subsequent solving of this amended formula~$F^{(1)} := F \cup \Clauses$ by an SLS solver.
Our empirical evaluations show that this distribution is long-tailed.
We want to stress the fact that studies on the runtime distribution of algorithms are quite sparse %
even though
knowledge of the runtime distribution of an algorithm is extremely valuable:
\textcolor{lipicsGray}{\textbf{(1)}}
Intuitively speaking, if the distribution is %
long-tailed, one knows there is a risk of ending in the tail and experiencing very long runs;
simultaneously, the knowledge that the time the algorithm used thus far is in the tail of the distribution can be exploited to restart the procedure (and create a new logically equivalent instance~$F^{(2)}$).
We will prove this statement in a rigorous manner 
for all long-tailed algorithms.
\textcolor{lipicsGray}{\textbf{(2)}}
Given the distribution of an algorithm's sequential runtime, it was shown 
how to predict and quantify the algorithm's expected speedup due to parallelization~\cite{ATCTC13UsingSequentialRuntimeDistributions}. %
\textcolor{lipicsGray}{\textbf{(3)}}
If the distribution of hardness is known, 
experiments with 
few
instances can lead to 
parameter estimations
of the underlying distribution~\cite{FRV97SummarizingCSPHardness}.	
\textcolor{lipicsGray}{\textbf{(4)}}~Knowledge of the distribution can help 
compare competing algorithms: %
one can test if the difference in the means of 
algorithm runtimes is significant if the distributions are known~\cite{FRV97SummarizingCSPHardness}.

\subsection{Our Contributions}

Our contributions consist of an empirical as well as theoretical part, specified below.

\subparagraph*{Statistical Runtime/Hardness Distribution Analysis.}

By conducting a plethora of experiments (total CPU time $80$ years) and using several statistical tools for the analysis of 
empirical
distributions,
we conjecture that
\Alfa{} equipped with SLS solvers based on Schöning's Random Walk Algorithm~\cite{Schoening02AProbabilisticAlgorithm}, \SRWA{} for short, follows a long-tailed distribution (Conjecture~\ref{conj:weak}).
The evidence obtained further suggests that this distribution is, in fact, lognormal (Conjecture~\ref{conj:strong}).
We measure the goodness-of-fit of our results over the whole domain using the $\chi^2$-statistic.

\subparagraph*{Restarts Are Useful For Long-Tailed Algorithms.}

Lorenz~\cite{Lorenz18RuntimeDistributions} has analyzed
the lognormal and the generalized Pareto distribution
for the usefulness of restarts and their optimal restart times.
Given that our Strong~Conjecture~\ref{conj:strong} holds, this result implies that restarts are useful for~\Alfa{}.
We will also show that this is the case if only the Weak~Conjecture~\ref{conj:weak} holds: We theoretically prove that restarts are useful for the class of algorithms exhibiting a long-tailed distribution.

\subsection{Related Work and Differentiation}

In~\cite{FRV97SummarizingCSPHardness}, the authors presented empirical evidence for the fact that the distribution of the effort (more precisely, the number of consistency checks) required for backtracking algorithms to solve constraint satisfaction problems randomly generated at the 50\,\% satisfiable point can be approximated by the Weibull distribution (in the satisfiable case) and the lognormal distribution (in the unsatisfiable case).
These results were later extended to a wider region around the 50\,\% satisfiable point~\cite{RF97StatisticalAnalysis}.
It should be emphasized that this study created all instances using the same generation model. This resulted in the creation of similar yet logically non-equivalent formulas.
We, however, will firstly use different models to rule out any influence of the generation model and secondly generate logically equivalent modifications~of~a base instance (see \refAlg{algo:main}).
This approach lends itself to the analysis of existing~SLS solvers~\cite{LW20OnTheEffectOfLearnedClauses}.
The major advantage is that the conducted work is not lost in the case~of a restart: only the logically equivalent instance could be \mbox{changed while keeping the current assignment}.

In~\cite{GSCK00HeavyTailedPhenomena}, the cost profiles of combinatorial search procedures were studied. 
The authors showed 
that 
they
are often characterized by %
the
Pareto-Lévy distribution %
and %
empirically demonstrated
how rapid randomized restarts can 
eliminate %
this tail behavior.
We will theoretically prove the effectiveness of restarts for the larger class of long-tailed distributions.

The paper~\cite{ATCTC13UsingSequentialRuntimeDistributions} studied the 
solvers \algoformat{Sparrow} and \algoformat{CCASAT}
and found that for randomly generated instances
the lognormal distribution is a good fit
for the runtime distributions.
For this, the Kolmogorov--Smirnov %
statistic $\sup_{t \in \R} | \hat{F}_n(t) - F(t) |$ was used.
Although the KS-test is very versatile, this comes with the disadvantage that its statistical power is rather low.
Clearly, the KS statistic is also nearly useless in the tails of a distribution: A high relative deviation of the empirical from the theoretical cumulative distribution function in either tail results in a very small absolute deviation.
It should also be remarked that the paper studies only few formulas in just two domains, 10 randomly generated and 9 crafted.
Our work will address both shortcomings in this paper: The $\chi^2$-test gives equal importance to the goodness-of-fit over the entire support; and various instance domain models (both theoretical and applied) are considered in this paper.

\begin{remark*}
	Unfortunately, the term heavy- or long-tailed distribution is not used consistently in the literature.
	We will follow~\cite{foss2011introduction} and use the notion given in~\refDef{def:long_tail}.
\end{remark*}

\section{Preliminaries}

We assume familiarity with terminologies such as Boolean variable, literal, clause, CNF formula, the SAT problem, and assignment flips in SLS solvers and refer the reader to \egcite~\cite{ST13SATProblem}.
We furthermore trust that the reader has basic knowledge of the proof system Resolution \cite{Blake37Canoncial,Robinson65Machine-oriented}.
\introduceterm{Stochastic local search} (SLS) solvers operate on complete assignments for a formula~$F$.
These solvers are started with a randomly generated complete initial assignment~$\alpha_0$.
If~$\alpha_0$ satisfies~$F$, a solution is found.
Otherwise, the SLS solver tries to find a solution by performing a random walk over the set of complete assignments for the underlying formula.
A formula~$F$ \introduceterm{logically implies} a clause~$C$ if every complete truth assignment which satisfies~$F$ also satisfies~$C$, for which we write $F \vDash C$. If $L$ is a set of clauses we write $F \vDash L$ if $F \vDash C$ for all $C \in L$.

\begin{definition}
	\label{def:RestartsUseful}
	Let $X$ be a random variable 
	for
	the runtime of an 
	SLS 
	algorithm~$\mathcal{A}$ on some input.
	For $t > 0$, the algorithm~$\mathcal{A}_t$
	is obtained by restarting~$\mathcal{A}$ after time~$t$ if no solution was found. %
	Restarts are \introduceterm{useful} if there is a $t > 0$ %
	such that
	\begin{align*}
		\expectation{X_t} < \expectation{X},
	\end{align*}
	where $X_t$ models the runtime of~$\mathcal{A}_t$. 
\end{definition}

\begin{definition}[\cite{norman1994continuous}]
	\label{def:cdf_quantile}
	Let $X$ be a real-valued random variable.
	\begin{itemize}
	\item
		Its \introduceterm{cumulative distribution function}
		(cdf) %
		is the function $F \colon \R \to [0,1]$ with 
			\[
			F(t) := \prob{X \leq t}. %
			\]
	\item
		Its \introduceterm{quantile function} $Q \colon (0,1) \to \R$ is given by $Q(p) := \inf \setdescr{t \in \R}{F(t) \geq p}$.
	\item
		A non-negative, integrable function~$f$ such that 
		$F(t) = \integralLow[u]{-\infty}{t}{f(u)}$
		is 
		called 
		\introduceterm{probability density function} (pdf) of~$X$.
	\end{itemize}
\end{definition}

\begin{definition}[\cite{Wicksell17OnLogarithmicCorrelation}]
	An absolutely continuous, positive random variable~$X$ is \introduceterm{(three-parameter) lognormally distributed} with parameters 
	$\sigma^2 > 0$, $\gamma > 0$, and $\mu \in \R$, 
	if $\log(X - \gamma)$ is normally distributed with mean $\mu$ and variance $\sigma^2$. In the following, we refer to $\sigma$ as the \introduceterm{shape}, $\mu$ as the \introduceterm{scale}, and $\gamma$ as the \introduceterm{location parameter}.
\end{definition}

\begin{definition}
	Let $X_1, \dots, X_n$ be independent, identically distributed real-valued random variables
	with realizations $x_i$ of $X_i$.
	Then the \introduceterm{empirical cumulative distribution function} (\introduceterm{ecdf}) of the sample $(x_1, \dots, x_n)$ is defined as
		\[
		\hat{F}_n(t) := \frac{1}{n} \sum_{i=1}^n \Indikator{\set{x_i \leq t}}, \quad t \in \R,
		\]	
	where $\Indikator{A}$ is the indicator of event $A$.
\end{definition}

\section{Design of the Adjusted Logical Formula Algorithm Alfa}
\label{sec:DesignOfAlfa}

Our SLS solver \Alfa{} (Adjusted logical formula algorithm) receives a satisfiable formula~$F$ as input.
The algorithm then proceeds by adding to~$F$ a set~$\Clauses$ of logically generated clauses.
It finally calls an SLS solver to solve the clause set~$F \cup \Clauses$.

\begin{algorithm}[htb]
	\textbf{Input:} Boolean formula $F$, \textbf{Promise:} $F \in \SAT$\\
	\BlankLine
	Generate \textbf{randomly} a set $\Clauses$ of clauses such that $F \vDash \Clauses$\\
	Call \algoformat{SLS}$(F \cup \Clauses)$ for some SLS solver~\algoformat{SLS}
	\caption{\Alfa{} acts as a base algorithm that can use different SLS algorithms.}
	\label{algo:main}
\end{algorithm}

Definition~\ref{def:width-w-res} is used
in Algorithm~\ref{algo:res} as a natural way to sample a set~$\Clauses$ of logically equivalent clauses with respect to a base instance~$F$.

\begin{algorithm}[htb]
	
	\SetKw{KwWithProb}{with probability}
	\SetKw{KwDo}{do}	
	\SetKwFunction{Shuffle}{Shuffle}
	
	\textbf{Input:} Boolean formula $F$, integer $w$, probability $p \in (0,1]$, Boolean \textit{shuffle}\\
	\BlankLine
	
	\ForEach{$R \in \operatorname{Res}_w^{\ast}(F) \setminus F$}{
		\KwWithProb $p$ %
		\KwDo
		$L := L \cup \set{R}$
	}
	
	\lIf{shuffle}{\textbf{return} \Shuffle{$L$} \textbf{else} \textbf{return} $L$}
	
	\caption{Generation of the random set $\Clauses$ with resolution}
	\label{algo:res}
\end{algorithm}

\begin{definition}
	\label{def:width-w-res}
	Let $F$ be a clause set, and $w$ be a positive integer. We define the operator
	\[
	\operatorname{Res}_w(F) := F \cup \Setdescr{R}{R \text{ is a resolvent of two clauses in } F \text{ and } |R| \leq w}.
	\]
	Also, we inductively define $\operatorname{Res}_w^{0}(F) := F$ and
	\[
	\operatorname{Res}_w^{n+1}(F) := \operatorname{Res}_w \! \big( \operatorname{Res}_w^{n}(F) \big), \text{ for } n \geq 0.
	\]
	Finally, we set	
	\[
	\operatorname{Res}_w^{\ast}(F) := \bigcup_{n \geq 0} \operatorname{Res}_w^{n}(F).
	\]
\end{definition}

\section{Empirical Evaluation}

\subsection{Experimental Setup, Instance Types, and Solvers Used}

Hoos and Stützle~\cite{HS98EvaluatingLasVegas} introduced the concept of \emph{runtime distribition} to characterize the cdf of Las Vegas algorithms, where the runtime can vary from one execution to another, even with the same input.
To obtain enough data for a fitting of such a distribution, for each base instance $F$ we created 5000 modified instances $F^{(1)}, \dots, F^{(5000)}$ by generating resolvent sets $\Clauses^{(1)}, \dots, \Clauses^{(5000)}$ %
each by using \refAlg{algo:res} with $w=4$ and a value of $p$ such that the expected number of resolvents being added was $\frac{1}{10} |F|$.
Note that we also conducted a series of experiments to rule out the influence of $p$ on our results.
Each of these modified instances was solved 100 times, each time using a different seed.
For $i = 1, \dots, 5000$ and $j = 1, \dots, 100$ we thus obtained the values $\mathsf{flips}_{S}(F^{(i)}, s_j)$
indicating how many flips were used to solve the modified instance~$F^{(i)}$ with solver~$S$ when using the seed~$s_j$.
Next, we calculated the mean number of flips $\mathsf{mean}_S(F^{(i)}) := \frac{1}{100} \sum_{j=1}^{100} \mathsf{flips}_{S}(F^{(i)}, s_j)$ required to solve $F^{(i)}$ with solver~$S$ whose hardness distribution we are going to analyze.

All experiments were performed on bwUniCluster 2.0 and 
three 
local servers.
Sputnik~\cite{VLSKK15Sputnik} was used to distribute the computation and to parallelize 
the 
trials.
Due to the heterogeneity of the 
computer 
setup, measured runtimes are not directly comparable to each other.
Consequently, we instead measured the number of variable flips performed by the SLS solver.
This is a hardware-independent performance measure with the benefit that it can also be analyzed 
theoretically.
To give an indication of how flips relate to wall-clock time, 
one million flips take about one second of computing time on one of our servers.
To give an idea of the computational effort involved, obtaining
the data for the ecdf of a 100 variable base instance with \SRWA{} took an average of 17,193,517 seconds ($\approx 199$~days) when unparallelized.
This clearly prohibited examining instances having a number of variables currently being routinely solved by the state-of-the-art SLS algorithms.
For the experiments the following instance types were used:

\begin{enumerate}
	\item \textbf{\textcolor{lipicsGray}{Hidden Solution:}}		
	We implemented the CDC algorithm~\cite{BC18UsingAlgorithmConfigurationTools, BHLRTWZ02Hiding} in~\cite{LW21SourceCodeOfConcealSATgen}
	to generate instances with a hidden solution. 
	For this, 
	at the beginning,
	a complete assignment~$\alpha$ is specified to ensure the generated formula's satisfiability.
	Then, repeatedly a randomly generated clause~$C$ is added to the formula with a weighted probability $p_i$ depending on the number~$i$ of correct literals in $C$ with respect to $\alpha$. 
	We included this type of instances because SLS solvers struggle to solve such instances. Experiments like these might be beneficial to find theoretical reasons for this behavior.	
	
	\item \textbf{\textcolor{lipicsGray}{Hidden Solution With Different Chances:}}			
	We also created sets of formulas with different underlying $p_i$ values to rule out the influence of these.

	\item \textbf{\textcolor{lipicsGray}{Uniform Random:}}	
	To generate
	uniform, random $k$-SAT instances with $n$ variables and $m$ clauses, each
	clause is generated by sampling $k$~literals uniformly and independently. 
	Using Gableske's kcnfgen~\cite{kcnfgen},
	we generated formulas with $n \in \set{50,60,70,80,90}$ variables and
	a clause-to-variable ratio~$r$ close to the \emph{satisfiability threshold}~\cite{MMZ06ThresholdValues} of $r \approx 4.267$.
	We checked each instance 
	with \texttt{Glucose3}~\cite{AS09Glucose,ES03MiniSat}
	for satisfiability
	until we had 5 formulas of each size.

	\item \textbf{\textcolor{lipicsGray}{Factoring:}} 
	These formulas encode the factoring problem in the interval $\set{128, \dots, 256}$ and were generated with~\cite{Diemer21GenFactorSat}.

	\item \textbf{\textcolor{lipicsGray}{Coloring:}}
	These formulas assert that a 
	graph 
	is colorable with 
	3 colors.
	We generated 
	these formulas, 
	using~\cite{LENV17CNFgen}, over random graphs with $n$ vertices and $m = 2.254n$ edges in expectation,
	which is slightly below the \introduceterm{non-colorability threshold}~\cite{KaporisKS00}.
	We
	obtained 32~satisfiable instances in 150 variables.

\end{enumerate}

Our experiments investigated leading SLS solvers where the dominating component is based on the random walk procedure proposed in~\cite{Schoening02AProbabilisticAlgorithm}.
In this paper, Schöning's Random Walk Algorithm \SRWA{} was introduced, which is one of the solvers we used.
The \probSAT{} solver family~\cite{BalintImplementationOfProbSAT} is based on this approach.
One of these solvers won the random track of the SAT competition 2013~\cite{SAT13}.
Another advancement of \algoformat{SWRA} was implemented as \YalSAT{}~\cite{YalSAT}, which won the random track of the SAT competition 2017~\cite{SAT17}.
These performances and similarities were reasons for choosing \SRWA{}, \probSAT{}, and \YalSAT{} as SLS solvers for this paper. The connection to \GapSAT{}~\cite{LW20OnTheEffectOfLearnedClauses} is another case in point.

We excluded the solvers \algoformat{DCCAlm}~\cite{LCS16DCCAlm} and \algoformat{CSCCSat}~\cite{LCWS16CSCCSat} (combining \algoformat{FrwCB}~\cite{LCWS13FRW} and \algoformat{DCCASat}~\cite{LCWS14DoubleConfiguration}) as all of these 
depend on 
heuristics (like \textsf{CC}, \textsf{BM}, \textsf{CSDvars}, \textsf{NVDvars}, \textsf{SDvars}) that ultimately reduce the probabilistic nature when choosing the next variable to flip.

For \SRWA{} we conducted most of our experiments: All instance types were tested, including different change values for the generation of the hidden solution.
For \probSAT{}, 55 hidden solution instances with $n \in \set{50,100,150,200,300,800}$ were used.
Since \YAL{} can be regarded as a \probSAT{} derivate, 
we tested \YAL{} with 10 hidden solution instances with 300 variables each.

\subsection{Experimental Results and Statistical Evaluation}

The goal of this section is to explore the scenarios described above in more detail. We are particularly interested in how the hardness of an instance changes when logically equivalent clauses are added in the manner described above. To characterize this effect as accurately as possible, studying the ecdf is the most suitable method for this purpose. In turn, the ecdf can be described using well-known distribution types such as \eg the normal distribution. In the following, we shall demonstrate that the three-parameter lognormal distribution, in particular, provides an exceptionally accurate description of the runtime behavior, and this is true for all considered problem domains and all solvers. The results are so compelling that we ultimately conjecture that the runtimes of \Alfa{}-type algorithms all follow a lognormal distribution, regardless of the considered problem domain. 

To illustrate this point, we first demonstrate our approach using two base instances. The first one is a factorization instance that was solved by \SRWA{}. The second instance has a hidden solution and was solved by \probSAT{}. For later reference, we refer to the first instance as~$A$ and to the second instance as~$B$. As described above, we obtain 5000~samples for each base instance. Using these data points, we estimate the lognormal distribution's three parameters by applying the maximum likelihood method (see~\cite{AllData}). After that, one can visually evaluate the suitability of the fitted lognormal distribution for describing the data. A useful method of visualizing the suitability is to plot the ecdf and the fitted cdf on the same graph.

\begin{figure}[htb]
	\centering
	\begin{subfigure}[b]{.465\linewidth}
		\centering
		\definecolor{mycolor1}{rgb}{0.00000,0.44700,0.74100}%
\definecolor{mycolor2}{rgb}{0.85000,0.32500,0.09800}%

\begin{tikzpicture}

	\begin{axis}[
		width=5.0in*0.6,
		height=3.6in*0.6, %
		tick label style={font=\scriptsize},
		tick align=inside,
		tick pos=left,
		x grid style={white!69.02!black},
		xmajorgrids,
		xmin=0, xmax=600000000,
		xtick style={color=black},
		y grid style={white!69.02!black},
		ymajorgrids,
		ymin=0.0, ymax=1.0,
		ytick style={color=black},
		xlabel={\footnotesize{flips}},
		legend pos=south east,
	]
	\addplot [
		semithick, color=mycolor1, line width=1.2pt, mark options={solid}
	] table {tikz/tables/factor_empirical_cdf.csv};
	\addlegendentry{\footnotesize{empirical}}
	\addplot [
		color=black, dashed, black, line width=1.2pt
	] table {tikz/tables/logn_fit_cdf_lin.csv};
	\addlegendentry{\footnotesize{fitted}}
	\end{axis}
\end{tikzpicture}
	\end{subfigure}%
	\hfill%
	\begin{subfigure}[b]{.465\linewidth}
		\centering
		\definecolor{mycolor1}{rgb}{0.00000,0.44700,0.74100}%
\definecolor{mycolor2}{rgb}{0.85000,0.32500,0.09800}%

\begin{tikzpicture}

	\begin{axis}[
		width=5.0in*0.6,
		height=3.6in*0.6,
		tick label style={font=\scriptsize},
		tick align=inside,
		tick pos=left,
		x grid style={white!69.02!black},
		xmajorgrids,
		xmin=0, xmax=15000,
		xtick style={color=black},
		y grid style={white!69.02!black},
		ymajorgrids,
		ymin=0.0, ymax=1.0,
		ytick style={color=black},
		xlabel={\footnotesize{flips}},
		yticklabels={,,},
		legend pos=south west,
	]
	\addplot [
		semithick, color=mycolor1, line width=1.2pt, mark options={solid}
	] table {tikz/probsat/empirical_cdf.csv};
	\addlegendentry{\footnotesize{empirical}}
	\addplot [
		color=black, dashed, black, line width=1.2pt
	] table {tikz/probsat/logn_fit_cdf_lin.csv};
	\addlegendentry{\footnotesize{fitted}}
	\end{axis}
\end{tikzpicture}
	\end{subfigure}%
	\caption{%
		The ecdf and fitted cdf of the hardness distribution %
		of instance $A$ (left) and $B$ (right).
	}
	\label{fig:cdf-ecdf}
\end{figure}
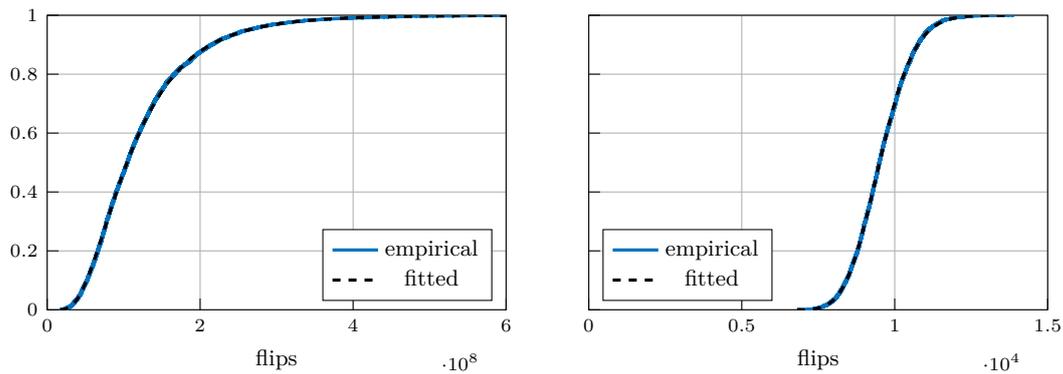
Such a comparison is illustrated in \refFig{fig:cdf-ecdf} for the two instances~$A$ and~$B$. In both cases, no difference between the empirical data of the ecdf and the fitted distribution can be detected visually. In other words, the absolute error between the predicted probabilities from the fitted cdf versus the empirical probabilities from the ecdf is minuscule. Even though these are only two examples, it should be noted that these two instances are representative of the behavior of the investigated algorithms. Hardly any deviation could be observed in this plot type for all instances and all algorithms. All data is published under~\cite{AllData}.

For the analysis, however, one should not confine oneself to this plot type. Although absolute errors can be observed easily, relative errors are more difficult to detect. Such a relative error may have a significant impact when used for decisions such as restarts. To illustrate this point, suppose that the true probability of a run of length $\ell$ is $0.0001$. In contrast, the probability estimated based on a fit is $0.001$. As can be seen, the absolute error of $0.0009$ is small, whereas the relative error of $10$ is large. If one were to perform restarts after $\ell$ steps, the actual expected runtime would be ten times greater than the estimated expected runtime. Thus, the erroneous estimate of that probability would have translated into an unfavorable runtime. This example should illustrate the importance of checking the tails of a distribution for errors as well. 

The left tail, \ie the probabilities for 
very
small values, can be checked visually by plotting the ecdf and 
fitted cdf with both axes logarithmically scaled.
Thereby, the probabilities for extreme events (in this case, especially easy instances) can be measured accurately. 
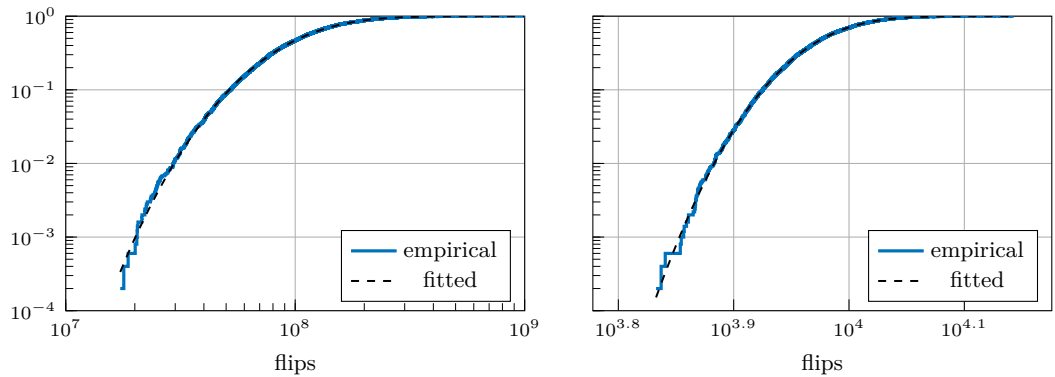
\begin{figure}[htb]
	\centering
	\begin{subfigure}[t]{.465\linewidth}
		\centering
		\definecolor{mycolor1}{rgb}{0.00000,0.44700,0.74100}%
\definecolor{mycolor2}{rgb}{0.85000,0.32500,0.09800}%

\begin{tikzpicture}

	\begin{axis}[
	width=5.0in*0.6,
	height=3.6in*0.6,
		tick label style={font=\scriptsize},
		log basis x={10},
		log basis y={10},
		tick align=inside,
		tick pos=left,
		x grid style={white!69.02!black},
		xmajorgrids,
		xmin=10000000, xmax=1000000000,
		xmode=log,
		xtick style={color=black},
		y grid style={white!69.02!black},
		ymajorgrids,
		ymin=0.0001, ymax=1.0,
		ymode=log,
		ytick style={color=black},
		xlabel={\footnotesize{flips}},
		legend pos=south east
	]
	\addplot [
		semithick, color=mycolor1, line width=1.2pt, mark options={solid}
	] table {tikz/tables/factor_empirical_cdf.csv};
	\addlegendentry{\footnotesize{empirical}}
	\addplot [
		color=black, dashed, black, line width=0.75pt
	] table {tikz/tables/logn_fit_cdf_geom.csv};
	\addlegendentry{\footnotesize{fitted}}
	\end{axis}
\end{tikzpicture}
	\end{subfigure}%
	\hfill%
	\begin{subfigure}[t]{.465\linewidth}
		\centering
		\definecolor{mycolor1}{rgb}{0.00000,0.44700,0.74100}%
\definecolor{mycolor2}{rgb}{0.85000,0.32500,0.09800}%

\begin{tikzpicture}

	\begin{axis}[
	width=5.0in*0.6,
	height=3.6in*0.6,
		tick label style={font=\scriptsize},
		log basis x={10},
		log basis y={10},
		tick align=inside,
		tick pos=left,
		x grid style={white!69.02!black},
		xmajorgrids,
		xmin=6000, xmax=15000,
		xmode=log,
		xtick style={color=black},
		y grid style={white!69.02!black},
		ymajorgrids,
		ymin=0.0001, ymax=1.0,
		ymode=log,
		ytick style={color=black},
		xlabel={\footnotesize{flips}},
		yticklabels={,,},
		legend pos=south east
	]
	\addplot [
		semithick, color=mycolor1, line width=1.2pt, mark options={solid}
	] table {tikz/probsat/empirical_cdf.csv};
	\addlegendentry{\footnotesize{empirical}}
	\addplot [
		color=black, dashed, black, line width=0.75pt
	] table {tikz/probsat/logn_fit_cdf_geom.csv};
	\addlegendentry{\footnotesize{fitted}}
	\end{axis}
\end{tikzpicture}
	\end{subfigure}%
	\caption{%
		Logarithmically scaled ecdf and fitted cdf of instances $A$ (left) and $B$ (right).
	}
	\label{fig:log_cdf-log_ecdf}
\end{figure}

The two instances~$A$ and~$B$ are being examined in this manner in \refFig{fig:log_cdf-log_ecdf}. As can be observed, the lognormal fit accurately predicts the probabilities associated with very short runs. For the other instances, lognormal distributions were mostly also able to accurately describe the probabilities for short runs. However, the behavior of the ecdf and the fitted lognormal distribution differed very slightly in a few instances. 

Lastly, the probabilities for particularly hard instances should also be checked.  Any mistakes in this area could lead to underestimating the likelihood of encountering an exceptionally hard instance. For analyses of this type, the survival function $S$ is a useful tool; if $F$ is the cdf, $S(x):=1-F(x)$. Therefore, the survival function's value $S(x)$ represents the probability that an instance is (on average) harder than $x$ in our case. If we plot the empirical survival function, \ie $\hat{S}_n(x) := 1-\hat{F}_n(x)$, and the fitted survival function together on a graph with logarithmically scaled axes, we can easily detect errors in the right tail.

\begin{figure}[htb]
	\centering
	\begin{subfigure}[t]{.465\linewidth}
		\centering
		\definecolor{mycolor1}{rgb}{0.00000,0.44700,0.74100}%
\definecolor{mycolor2}{rgb}{0.85000,0.32500,0.09800}%

\begin{tikzpicture}

	\begin{axis}[
	width=5.0in*0.6,
	height=3.6in*0.6,
		tick label style={font=\scriptsize},
		log basis x={10},
		log basis y={10},
		tick align=inside,
		tick pos=left,
		x grid style={white!69.02!black},
		xmajorgrids,
		xmin=10000000, xmax=1000000000,
		xmode=log,
		xtick style={color=black},
		y grid style={white!69.02!black},
		ymajorgrids,
		ymin=0.0001, ymax=1.0,
		ymode=log,
		ytick style={color=black},
		xlabel={\footnotesize{flips}},
		legend pos=south west,
	]
	\addplot [
		semithick, color=mycolor1, line width=1.2pt, mark options={solid}
	] table {tikz/tables/factor_empirical_survival.csv};
	\addlegendentry{\footnotesize{empirical}}
	\addplot [
		color=black, dashed, line width=0.75pt
	] table {tikz/tables/logn_fit_sf_geom.csv};
	\addlegendentry{\footnotesize{fitted}}
	\end{axis}
\end{tikzpicture}
	\end{subfigure}%
	\hfill%
	\begin{subfigure}[t]{.465\linewidth}
		\centering
		\definecolor{mycolor1}{rgb}{0.00000,0.44700,0.74100}%
\definecolor{mycolor2}{rgb}{0.85000,0.32500,0.09800}%

\begin{tikzpicture}

	\begin{axis}[
	width=5.0in*0.6,
	height=3.6in*0.6,
	tick label style={font=\scriptsize},
	log basis x={10},
	log basis y={10},
	tick align=inside,
	tick pos=left,
	x grid style={white!69.02!black},
	xmajorgrids,
	xmin=6000, xmax=15000,
	xmode=log,
	xtick style={color=black},
	y grid style={white!69.02!black},
	ymajorgrids,
	ymin=0.0001, ymax=1.0,
	ymode=log,
	ytick style={color=black},
	xlabel={\footnotesize{flips}},
	yticklabels={,,},
	legend pos=south west
	]
	\addplot [
		semithick, color=mycolor1, line width=1.2pt, mark options={solid}
	] table {tikz/probsat/empirical_sf.csv};
	\addlegendentry{\footnotesize{empirical}}
	\addplot [
		color=black, dashed, line width=0.75pt
	] table {tikz/probsat/logn_fit_sf_geom.csv};
	\addlegendentry{\footnotesize{fitted}}
	\end{axis}
\end{tikzpicture}
	\end{subfigure}%
	\caption{%
		Logarithmically scaled empirical survial function and fitted survival function of instances~$A$ (left) and $B$ (right).
	}
	\label{fig:sf-esf}
\end{figure}
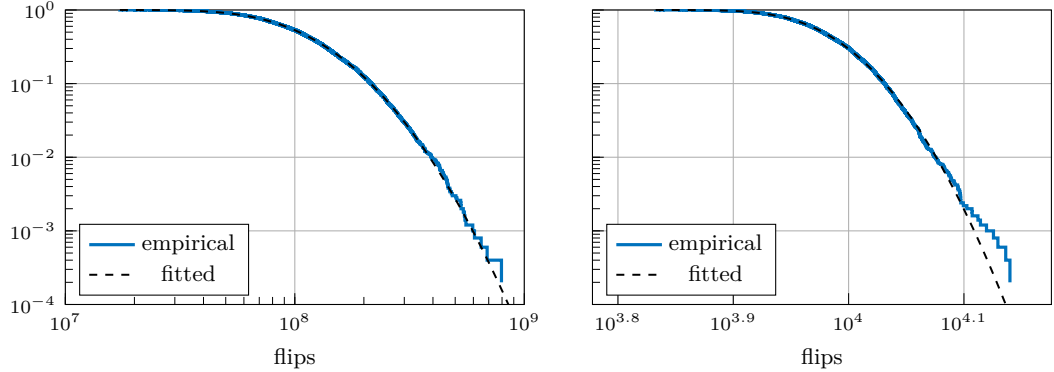

\refFig{fig:sf-esf} illustrates this type of plot for the instances~$A$ and~$B$. Here, there is a discernible deviation between $A$ and $B$. While for $A$, the lognormal fit provides an accurate description of the probabilities for long runs, in the case of $B$, the empirical survival function seems to approach $0$ somewhat slower than the lognormal estimate. In the vast majority of cases, these extreme value probabilities are accurately reflected by the lognormal fit. In most other cases, the empirical survival function approaches 0 more slowly than the lognormal fit. Thus, in these cases, the likelihood of encountering an exceptionally hard instance is underestimated.

So far, we discussed the behavior of lognormal fits based on this visual inspection. Altogether, we concluded that lognormal distributions seem to be well suited for describing the data. Next, we shall concretize this through a statistical test. To be more precise, we apply the $\chi^2$-test as a goodness-of-fit test for each base instance.
For each such instance, the fitted lognormal distribution used the $5000$ data points, and afterwards, the $\chi^2$-test statistic is computed. Subsequently, the probability that such a value of the test statistic occurs under the assumption of the so-called null hypothesis is determined. We will refer to this probability as the $p$-value. In our case, the null hypothesis is the assumption that the data follow a lognormal distribution. If the fit is poor, then a small $p$-value will occur. If $p$ is sufficiently small, the null hypothesis is rejected. We reject the null hypothesis if \mbox{$p<0.05$}. 

Two more remarks are due on this matter. First, from a high $p$-value, one cannot 
prove
that the assumption that the data are lognormally distributed is correct.
However, we use a sufficiently high $p$-value as a heuristic whether this assumption is reasonable.

Secondly, there is an obstacle that complicates statistical analysis by this method. As described,
each of the $5000$~data points is obtained by first sampling $100$~runtimes of the corresponding instance and then calculating the mean. 
This means that we do not work with the actual expected values but only estimates. In other words, this implies that our data is noisy. The greater the variance in the respective instance, the greater the corresponding noise. If one were to apply the $\chi^2$-test to this noisy data, some cases would be incorrectly rejected, especially if the variance is large. To overcome this limitation, we additionally use a bootstrap-test, which is based on Cheng~\cite{cheng2017non}. This test is presented %
in \refAlg{algo:bootstrap}.

\begin{algorithm}[htb]
	\textbf{Input:} (noisy) random sample $\mathbf{y}=(y_1, y_2, \dots, y_n)$, integer $N$, significance $\alpha\in (0,1)$\\
	\BlankLine
	
	$\hat{\theta}\leftarrow \textnormal{MLE}(\mathbf{y}, F)$, lognormal maximum likelihood estimation, $F$ is the lognormal cdf\\
	$X^2 \leftarrow \textnormal{ChiSquare}(\mathbf{y}, \hat{\theta})$, Chi-squared goodness of fit test statistic\\
	
	\SetKw{KwWithProb}{with probability}
	\SetKw{KwDo}{do}
	
	\SetKwFunction{Shuffle}{Shuffle}
	
	\For{$j=1$ to $N$}{
		$\mathbf{y}' \leftarrow (y'_1, \dots, y'_n)$, where all $y'_i$ are i.\,i.\,d.\ samples from the fitted lognormal\\
		 \phantom{$\mathbf{y}' \leftarrow (y'_1, \dots, y'_n)$,} distribution with parameters $\hat{\theta}$\\
		$\mathbf{y}' \leftarrow \mathbf{y}' + \textnormal{noise}$, where noise is sampled from an $n$-dimensional normal distribution\\
		$\hat{\theta}'\leftarrow \textnormal{MLE}(\mathbf{y}', F)$\\
		$X^2_j \leftarrow \textnormal{ChiSquare}(\mathbf{y}', \hat{\theta}')$\\
	}
	Let $X^2_{(1)}\leq X^2_{(2)}\leq \dots \leq X^2_{(N)}$ be the sorted test statistics.\\
	\lIf{$X^2_{( \lfloor (1-\alpha)\cdot N \rfloor )}<X^2$}{reject \textbf{else} accept}

	\caption{Bootstrap-test for noisy data}
	\label{algo:bootstrap}
\end{algorithm}

Briefly summarized, this test simulates how our data points were generated, assuming the null hypothesis. For this purpose, particular attention should be paid to how the test sample is rendered noisy. Owing to the central limit theorem, it is reasonable to assume that the initial data's sample mean originates from a normal distribution around the true expected value.  We use this assumption in the bootstrap-test using a noise signal drawn from a normal distribution with expected value~$0$. The variance of this normal distribution is determined from the initial data and divided by $100$ (cf.\ central limit theorem). %

If there is a large difference between the respective $p$-values of the $\chi^2$- and bootstrap-tests, this suggests that the variance in the initial data is too high and that the number of samples used to calculate the sample mean should be increased. However, this case has only occurred twice, and we explicitly indicate it later. In all other cases, the $p$-values of the $\chi^2$- and the bootstrap-tests were similar.
To illustrate this point, let us consider the $p$-values of the two instances~$A$ and~$B$. The excellent representation of the runtime behavior over the entire support on instance~$A$ is also reflected in the two $p$-values. Using the $\chi^2$-test, one obtains a $p$-value of approximately $0.783$, and using the bootstrap-test, one obtains a $p$-value of $0.76$. Since these two $p$-values are both above $0.05$, we conclude that the assumption that lognormal distributions describe the runtimes is reasonable.

In contrast, for instance $B$, we observed that while a lognormal distribution accurately describes the main part %
of the distribution,
the probabilities for extremely long runs are inadequately represented. This observation is again reflected in the respective $p$-values. The $\chi^2$-test yields a $p$-value of 
$\approx 0.008$, and the bootstrap-test yields a $p$-value of $0.013$. Since both $p$-values are
below 
$0.05$, the assumption that these data originate from a lognormal distribution is rejected.

Overall, this example is intended to demonstrate that these two statistical tests can show an inadequate fit, even if the problems only arise for extreme values. Second, it should demonstrate that the $p$-values of the two tests are generally similar; if the $p$-values differ significantly, then more samples should be used to calculate the sample mean.

We now proceed by considering the adequacy of lognormal distributions for describing \SRWA{} runtimes.  The results of the statistical analysis are reported in \refTab{tab:stat} and can be found in~\cite{AllData}.
\begin{table}[htb]
	\begin{tabular}{l| c c c c c || c}
		& hidden & different chances & uniform & factoring & coloring & total\\ \hline
		rejected & 0 & 2 & 1 & 2 & 0 & 5\\
		$\#$ of instances & 20 & 120 & 25 & 33 & 32 & 230
	\end{tabular}
	\caption{Statistical goodness-of-fit results for \Alfa{}+\SRWA{} runtimes over various problem domains. The \emph{rejected} row contains the number of instances where the lognormal distribution is not a good fit according to the $\chi^2$-test at a significance level of $0.05$. To put these results into perspective, the second row contains the total number of instances of each domain. Out o a total of 230 instances, 5 got rejected.}
	\label{tab:stat}
\end{table}

The first line in the table represents for how many instances the statistical tests rejected the lognormal distribution hypothesis. The second line indicates how many instances have been checked in total. 
It should be noted that the same number of instances was rejected by the $\chi^2$-test as was by the bootstrap-test.
Thus,
there is no need to distinguish between the tests here. It should also be mentioned that in statistical tests, there is always a possibility that a hypothesis will be rejected even though the null hypothesis holds (type 1 error).
At a significance level of \mbox{$0.05$}, this probability is at the most $5\%$.
Accordingly, the total of $5$~rejected instances may be attributed to 
so-called
\mbox{type~1 errors}. This statement is also supported by the fact that no exceptionally low $p$-value was observed, \ie no $p$-value that is unusual for a total of $230$~samples.

For \probSAT{}, on the other hand, the situation appears to be different. The results are summarized in \refTab{tab:stat_probSAT} and~\cite{AllData}.
\begin{table}[htb]
	\begin{tabular}{l| c c c c c c || c}
		number of variables & 50 & 100 & 150 & 200 & 300 & 800 & total \\\hline
		rejected & 2 & 2 & 1 & 0 & 2 & 0 & 7\\
		$\#$ of instances & 10 & 10 & 10 & 10 & 10 & 5 & 55
	\end{tabular}
	\caption{Goodness-of-fit results for \Alfa{}+\probSAT{} over various hidden solution instance sizes.
	The \emph{rejected} row contains the number of instances where the lognormal distribution is not a good fit according to the $\chi^2$-test at a significance level of $0.05$. To put these results into perspective, the second row contains the total number of instances of each instance size.}
	\label{tab:stat_probSAT}
\end{table}
The columns refer to the number of variables in the corresponding SAT instances.
The number of rejected instances is again identical regardless of whether the $\chi^2$- or bootstrap-test is applied. %
As can be seen, the lognormal distribution hypothesis was rejected for $7$ of the $55$~instances.
This number can no longer be accounted for by type 1 errors at a significance level of~$0.05$.
However, one can observe that the majority of rejected instances occur for a small number of variables.
If one were to consider only the instances from $150$~variables onwards, then the remaining rejected instances may be attributed to type 1 errors.
This raises the suspicion that there may be a limiting process, \ie that the lognormal distribution hypothesis is only valid for $n\rightarrow \infty$.

Lastly, a difference between the two static tests emerges for \YAL{}. According to the $\chi^2$-test, $2$ of the total $10$~instances are rejected. However, using the bootstrap-test, the lognormal distribution hypothesis is not rejected for any instance. Therefore, one cannot rule out the possibility that lognormal distributions are the natural model to describe the instances, but more experiments are required to make a more precise statement.

In summary, the presumption that lognormal distributions are the appropriate choice for describing runtimes has been reinforced for \SRWA{}. For \probSAT{}, this appears plausible at least above a certain instance size. Likewise, the choice of lognormal distributions also seems reasonable for \YAL{}. These observations lead us to the following conjecture.

\begin{conjecture}[Strong Conjecture]
	\label{conj:strong}
	The runtime of $\Alfa$ with $\algoformat{SLS} \in \set{\texttt{SRWA}, \texttt{probSAT}, \texttt{YalSAT}}$ follows a lognormal distribution.
\end{conjecture}

If this statement is true, then it would be intriguing in that one can infer how modifying the base instance affects the hardness of instances.
This effect is likely to be the result of generating models for lognormal distributions.
Just as the normal distribution is a natural model for the sum of \iid random variables, the lognormal distribution is a natural model for the multiplication of \iid random variables. Thus, one can hypothesize that each added clause exerts a small multiplicative effect on the instance's hardness. 

Simultaneously, the three parameters of the lognormal distribution also provide insight into how the hardness of the instance changes. For example, the location parameter $\gamma$ implies an inherent problem hardness that cannot be decreased regardless of the added clauses' choice. 
At the same time, $\gamma$ also serves as a numerical description for the value of this intrinsic hardness.
Using Bayesian statistics, it is possible to infer the parameters while the solver is running.
These estimated parameters can, for example, be used to schedule restarts.
This would lead to a scenario similar to that discussed in~\cite{RHK02RestartPolicies}.

Conjecture~\ref{conj:strong} is a strong statement.
However, %
a small deviation of the probabilities, for example, at the left tail, would render the strong conjecture invalid from a 
strict
mathematical point of view. Particularly, 
visual analyses revealed that the left tail's behavior, \ie for extremely short runs, is occasionally not accurately reflected by lognormal distributions. Conversely, the right tail, \ie the probabilities for particularly long runs, are usually either correctly represented by lognormal distributions or, occasionally, the corresponding probability approaches $0$~even more slowly. We, therefore, rephrase our conjecture in a %
weakened form. Our observations fit a class of distributions known as long-tail distributions defined purely in terms of their behavior at the right tail.

\begin{definition}[\cite{foss2011introduction}]
	\label{def:long_tail}
	A positive, real-valued random variable $X$ is \introduceterm{long-tailed}, if and only if
	\begin{align*}
		\forall x\in \Rpos:\, \prob{X>x} > 0
		\quad \quad \textnormal{and} \quad \quad
		\forall y\in \Rpos: \lim\limits_{x\rightarrow \infty} \frac{\prob{X > x+y}}{\prob{X>x}} = 1.
	\end{align*}
\end{definition}
\begin{conjecture}[Weak Conjecture]
	\label{conj:weak}
	The runtime of \Alfa{} with $\algoformat{SLS} \in \set{\texttt{SRWA}, \texttt{probSAT}, \texttt{YalSAT}}$ follows a long-tailed distribution.
\end{conjecture}

It should be noted that lognormal distributions have the long-tail property~\cite{foss2011introduction,nair2020fundamentals}. That is, if the Strong Conjecture holds, the Weak Conjecture is implied. The reverse is, however, not true. In the next section, we show an important consequence in case the Weak Conjecture holds.

\section{Restarts Are Useful For Long-Tailed Distributions}

If the Strong Conjecture holds, \ie if the runtimes are lognormally distributed, then restarts are useful~\cite{Lorenz18RuntimeDistributions}.
This section extends this result and mathematically proves that restarts are useful even if only the Weak Conjecture holds.
This will be achieved by showing that restarts are useful for long-tailed distributions.

A condition 
for the usefulness of restarts, as defined in %
Definition~\ref{def:RestartsUseful}, was 
proven
in~\cite{Lorenz18RuntimeDistributions}. %
We will show the result using this theorem that is restated below.

\begin{theorem}[\cite{Lorenz18RuntimeDistributions}]
	\label{theo:sufficient}
	Let $X$ be a positive, real-valued random variable having quantile function~$Q$, then restarts are useful if and only if there is a quantile $p \in (0,1)$ such that
	\begin{equation*}
		R(p,X) :=
		(1-p)\cdot \frac{Q(p)}{\expectation{X}}+\frac{\integralLow[u]{0}{p}{Q(u)}}{\expectation{X}}<p.
	\end{equation*}
\end{theorem}

Even if the quantile function and the expected value are unknown, 
$R(p,X)$
can be characterized for large values of~$p$.

\begin{lemma}
	\label{lem:exp_infinite_restart}
	Consider a positive, real-valued random variable~$X$ with pdf~$f$ and quantile function~$Q$ such that~\mbox{$\expectation{X}<\infty$}.
	Also, assume that the limit $\lim_{t\rightarrow\infty} t^2 \cdot f(t)$ exists.
	Then,
	\begin{align*}
		\lim\limits_{p\rightarrow 1} R(p,X)
		= \lim\limits_{p\rightarrow 1} \, \left((1-p)\cdot \frac{Q(p)}{\expectation{X}}+\frac{\integralLow[u]{0}{p}{Q(u)}}{\expectation{X}}\right)
		= 1.
	\end{align*}
\end{lemma}

\begin{proof}
	In the following, let~$F$ and~$f$ be the cdf and pdf of~$X$, respectively.
	We start by specifying the derivative of~$Q$ with respect to~$p$ as a preliminary consideration.
	From $F = Q^{-1}$ and the application of the inverse function theorem~\cite{rudin1964principles} it follows:
	\begin{align}
	\label{eq:quantil_deriv}
		Q'(p) \coloneqq \frac{\D}{\D p} Q(p)=\frac{1}{f\bigl(Q(p)\bigr)}. 
	\end{align}
	
	As the first step in our proof, we consider the limiting value of the second summand of~$R(p,X)$.
	This value can be determined by integration by substitution with~\mbox{$x=Q(u)$} followed by applying the change of variable method with~\mbox{$p=F(t)$}:
	\begin{align*}
		\lim\limits_{p\rightarrow 1}\frac{\integralLow[u]{0}{p}{Q(u)}}{\expectation{X}}
		=\lim\limits_{p\rightarrow 1}\frac{\integralLow{0}{Q(p)}{x\cdot f(x)}}{\expectation{X}}
		=\lim\limits_{t\rightarrow \infty} \frac{\integralLow{0}{t}{x\cdot f(x)}}{\expectation{X}} =1.
	\end{align*}
	The last equality 
	holds %
	because the 
	numerator
	matches the definition of the expected value.
	
	Next, we examine the limit
	of 
	$(1-p)Q(p)/\expectation{X}$.
	Since 
	$\lim_{p \to 1} (1-p) = 0$,
	the limit of \mbox{$(1-p)\cdot Q(p)/\expectation{X}$} needs to be examined 
	more closely.
	For this purpose, L'Hospital's rule is applied twice as well as the change of variable method with \mbox{$p=F(t)$} is used in the following:
	\begin{align*}
		\lim\limits_{p\rightarrow 1} (1-p)\cdot Q(p)
		= \lim\limits_{p\rightarrow 1} Q(p)^2\cdot f\bigl(Q(p)\bigr)
		= \lim\limits_{t\rightarrow \infty} t^2\cdot f(t).
	\end{align*}
	It is well-known that if \mbox{$\liminf_{t\rightarrow \infty} t^2 \cdot f(t) >0$} were to hold, then the expected value $\expectation{X}$ would be infinite (this statement is, for example, implicitly given in~\cite{foss2011introduction}).
	This would contradict the premise of the lemma; therefore, \mbox{$\liminf_{t\rightarrow \infty} t^2 \cdot f(t) = 0$}.
	Moreover, since, by assumption, $\lim_{t\rightarrow\infty} t^2 \cdot f(t)$ exists, we may conclude that 
	\[
	\lim\limits_{t\rightarrow\infty} t^2 \cdot f(t) = \limsup\limits_{t\rightarrow\infty} t^2 \cdot f(t) = \liminf\limits_{t\rightarrow\infty} t^2 \cdot f(t) = 0. \qedhere
	\]
\end{proof}

A frequently used tool for the description of distributions is the hazard rate function.

\begin{definition}[\cite{rausand2020system}]
	\label{def:hazard_rate}
	Let $X$ be a positive, real-valued random variable having cdf~$F$ and pdf~$f$.
	The \introduceterm{hazard rate function} $r \colon \Rpos \to \Rpos$ of $X$ is given by %
		\begin{align*}
			r(t) := \frac{f(t)}{1-F(t)}.
		\end{align*}
\end{definition}

In particular, there is an interesting relationship between the long-tail property and the hazard rate function's behavior.

\begin{restatelemmaorig}[\cite{nair2020fundamentals}]
	\label{lem:restateorig}
	\label{lem:appendixlemma}
	Let $X$ be a positive, real-valued random variable with hazard rate function~$r$ such that the limit \mbox{$\lim\limits_{t\rightarrow\infty}r(t)$} exists.
	Then $X$ is long-tailed if and only if \mbox{$\LIM{t} r(t)=0$}.
\end{restatelemmaorig}

\begin{proof}
	The appendix contains a proof of this lemma since the manuscript~\cite{nair2020fundamentals} was still unpublished at the time of writing.
\end{proof}

With the help of these preliminary considerations, we are now ready to show that restarts are useful for long-tailed distributions.

\begin{theorem}
	\label{theo:long_tail_restarts}
	Consider a positive, long-tailed random variable~$X$ with continuous pdf~$f$ and hazard rate function~$r$. 
	Also assume that either \mbox{$\expectation{X}=\infty$} holds or the limits \mbox{$\lim_{t\rightarrow\infty}r(t)$} and \mbox{$\lim_{t\rightarrow\infty} t^2 \cdot f(t)$} both exist.
	In both cases, restarts are useful for $X$.
\end{theorem}

\begin{proof}
	Let $F$ be the cdf and $Q$ the quantile function of $X$.
	We begin with the case 
	$\expectation{X} = \infty$.
	According to \refTheo{theo:sufficient}, restarts are useful if and only if
		\[
		(1-p)\cdot \frac{Q(p)}{\expectation{X}} + \frac{1}{\expectation{X}} \cdot \integralLow[u]{0}{p}{Q(u)} < p
		\]
	for some $p \in (0,1)$. However, if the expected value~$\expectation{X}$ is infinite, then the left side of this inequality is zero and the inequality is obviously satisfied. Hence, the statement follows.
	
	Secondly, we assume that $\expectation{X} < \infty$ and that
	both
	\mbox{$\Lim{t}r(t)$} and \mbox{$\Lim{t} t^2 \cdot f(t)$} exist. \refEq{eq:quantil_deriv} can now be used to calculate the following derivative:
	\begin{align*}
		\frac{\D}{\D p} \Big( R(p,X) - p \Big)
		= \frac{\D}{\D p} \left((1-p)\cdot \frac{Q(p)}{\expectation{X}}+\frac{\integralLow[u]{0}{p}{Q(u)}}{\expectation{X}}-p\right)
		= \frac{1-p}{\expectation{X}\cdot f\bigl(Q(p)\bigr)}-1.
	\end{align*}
	Consider the limit of this expression for $p\rightarrow 1$. Once again, the change of variable method is applied with $p=F(t)$, resulting in:
	\begin{align*}
		\lim\limits_{p\rightarrow 1} \frac{1-p}{\expectation{X}\cdot f\bigl(Q(p)\bigr)}-1 
		= \lim\limits_{t\rightarrow \infty} \frac{1-F(t)}{\expectation{X}\cdot f(t)}-1 
		= \lim\limits_{t\rightarrow \infty} \frac{1}{\expectation{X}\cdot r(t)}-1.
	\end{align*}
	By assumption, $X$ has a long-tail distribution and the limit of \mbox{$\Lim{t}r(t)$} exists. For this reason, \mbox{$\lim_{t\rightarrow \infty}r(t)=0$} follows as a result of \refLem{lem:appendixlemma}. Furthermore, since $\expectation{X} < \infty$ holds, we may conclude that  	
	\begin{align}
		\label{eq:long_tail_exp_infinity}
		\lim\limits_{p\rightarrow 1} \frac{1-p}{\expectation{X}\cdot f\bigl(Q(p)\bigr)}-1 
		= \lim\limits_{t\rightarrow \infty} \frac{1}{\expectation{X}\cdot r(t)}-1 = \infty.
	\end{align}
	The condition from \refTheo{theo:sufficient} can be rephrased in such a way that restarts are useful if and only if 
	$R(p,X) - p < 0$.
	According to \refLem{lem:exp_infinite_restart}, the left-hand side of this inequality approaches $0$ for $p\rightarrow 1$.
	However, as has been shown in \refEq{eq:long_tail_exp_infinity}, the derivative of 
	$R(p,X)-p$
	approaches infinity for $p\rightarrow 1$. These two observations imply that there is a $p\in (0,1)$ satisfying 
	$R(p,X) - p < 0$.
	Consequently, restarts are useful for~$X$.
\end{proof}

It should be noted that the conditions of this theorem are not restrictive since all naturally occurring long-tail distributions satisfy these conditions (see also~\cite{nair2020fundamentals}).

\begin{conjecture}[Corollary of the Weak Conjecture]
	Restarts are useful for \Alfa{} with $\textsf{SLS} \in \set{\texttt{SRWA}, \texttt{probSAT}, \texttt{YalSAT}}$.
\end{conjecture}

If \refCon{conj:weak}  is true, then this statement follows immediately by \refTheo{theo:long_tail_restarts}.

\section{Conclusion}

We have provided compelling evidence that the runtime of \Alfa{} follows a long-tailed or lognormal distribution.
According to~\cite{SMP11RestartStrategies}, the usefulness of restarts is a necessary, however not 
a sufficient, condition to obtain super-linear speedups by 
parallelization. Since we have shown that the necessary 
condition is (presumably) satisfied, this immediately raises the 
question of whether super-linear speedups are obtained by parallelizing 
\Alfa{}-type algorithms.

We additionally want to pose the question whether some of the Conjectures~\ref{conj:strong} or~\ref{conj:weak} can be theoretically proven.
A first line of attack would be to analyze the special case of an %
solver like \SRWA{} whose runtime was already theoretically analyzed.

The technique of analyzing the runtime distribution of \Alfa{} could be further developed to help better understand the behavior of CDCL solvers.
These kind of solvers heavily employ the technique of adding new clauses and deleting some clauses.
This can be thought of as solving a new logically equivalent formula of the base instance.

Preliminary results on the solvers excluded for heuristic reasons seem to suggest that the \Alfa{}-method forces the runtime of the base solver to exhibit a multimodal behavior. Thus, the lognormal distribution is not a good fit in this case. However, an initial visual inspection of the data indicates an even heavier tail.

\bibliography{xy_refEvidenceLongTails}

\section{Proof of~\refLem{lem:restateorig}}

This section proves the connection between long-tailed distributed random variables and their hazard rate functions of \refLem{lem:restateorig}, restated next. This restatement is presented in a stronger form 
by providing
an additional equivalent statement.

\begin{restatelemma}	
	Let $X$ be a positive, real-valued random variable with hazard rate function~$r$ such that the limit \mbox{$\lim_{t\rightarrow\infty}r(t)$} exists.
	Then, the following three statements are equivalent:
	\begin{enumerate}
		\item $X$ is long-tailed. \label{proof:long_tail_1}
		\item $\LIM{x} \integral[t]{x}{x+y}{r(t)} = 0$, \quantorsep $\forall y > 0$.
		\label{proof:long_tail_2}
		\item $\LIM{t} r(t)=0$.
		\label{proof:long_tail_3}
	\end{enumerate}
\end{restatelemma}

\begin{proof}[Proof of \refLem{lem:restateorig}]
	The proof is taken from~\cite{nair2020fundamentals}. A well-known property of the hazard rate function is:
	\begin{align}
		\prob{X > x} = \exp{\left(-\integralLow[t]{0}{x}{r(t)}\right)}.\label{eq:hazard_rate_connection_survival}
	\end{align}
	In the following, let $y \in \Rpos$ be any positive real number. First, we show the equivalence of the first and second statement.
	In line with the definition of long-tail distributions (see \refDef{def:long_tail}), we examine the quotient~$\frac{\prob{X > x+y}}{\prob{X>x}}$ for this purpose and apply the characterization from \refEq{eq:hazard_rate_connection_survival}:
	\begin{align*}
		\frac{\prob{X > x+y}}
		{\prob{X > x}} =
		\frac{\exp{\Big(-\integralLow[t]{0}{x+y}{r(t)}\Big)}}
		{\exp{\Big(-\integralLow[t]{0}{x}{r(t)}\Big)}} =
		\exp{\left(-\integral[t]{x}{x+y}{r(t)}\right)}.
	\end{align*}
	As one can readily infer from this equation, \mbox{$\Lim{x} \frac{\prob{X > x+y}}{\prob{X>x}}=1$} is true if and only if \mbox{$\Lim{x} \integralLow[t]{x}{x+y}{r(t)} = 0$} is true. Thus, the equivalence of the first and second statements has been demonstrated.
	
	In the next step, we show that the second statement implies the third statement, using logical contraposition to prove this. We, therefore, assume that \mbox{$\liminf\limits_{t\rightarrow \infty}r(t) > 0$} holds, or in other words, the following holds:
	\begin{align*}
		\exists C>0 \quantorsep \exists x_0 > 0 \quantorsep \forall x > x_0: \quantorsep
		r(x) \geq C.
	\end{align*}
	By exploiting this property, we are able to estimate the integral \mbox{$\integralLow[t]{x}{x+y}{r(t)}$}:
	\begin{align*}
		\forall x > x_0:\quantorsep
		\integral[t]{x}{x+y}{r(t)} \geq \integral[t]{x}{x+y}{C} = y\cdot C > 0.
	\end{align*}
	Thus it is shown that if the third statement does not hold, the second statement does not hold either.
	
	In the last step, we shall demonstrate that the third statement implies the second statement. Since $r(t)\geq 0$ holds for all $t\in \Rpos$, one may express \mbox{$\Lim{t} r(t)=0$} as follows:
	\begin{align*}
		\forall C>0 \quantorsep \exists x_0 > 0 \quantorsep \forall x > x_0:\quantorsep
		r(x) < C.
	\end{align*}
	Once again, we may exploit this property to estimate the integral~\mbox{$\integralLow[t]{x}{x+y}{r(t)}$}:
	\begin{align*}
		\forall x > x_0:\quantorsep
		\integral[t]{x}{x+y}{r(t)} < \integral[t]{x}{x+y}{C} = y\cdot C.
	\end{align*}
	Since this estimate holds for all $C>0$, this immediately yields \mbox{$\LIM{x} \integral[t]{x}{x+y}{r(t)} = 0$}.
\end{proof}

\end{document}